\documentclass[aps,pra,
amssymb,twocolumn
]{revtex4} 
\usepackage{graphicx}
\usepackage{amsmath}
\usepackage{amsfonts}
\usepackage{amssymb}
\usepackage{epsfig,libertine}
\usepackage[libertine]{newtxmath}
\usepackage{color}
\usepackage{dsfont, hyperref, xcolor}
\usepackage{comment}
\usepackage{enumerate}

\usepackage{amsthm} 
\newcommand{\Tr}[1]{\text{Tr}\left\{#1\right\}}

\newcommand{\bra}[1]{\langle#1\vert}
\newcommand{\ket}[1]{\vert#1\rangle}
\newcommand\braket[2]{\langle#1|#2\rangle}

\newtheorem{theorem}{Theorem}
\newtheorem{lemma}{Lemma}

\begin{document}

\title{Quantum advantage from negativity of a work quasiprobability distribution}

\author{Gianluca~Francica}
\address{Dipartimento di Fisica e Astronomia, Universit\`{a} di Padova, via Marzolo 8, 35131 Padova, Italy}

\date{\today}

\begin{abstract}
Quantum batteries can be charged by performing a work ``instantaneously'' in the limit of a large number of cells, achieving a so-called quantum advantage. In general, the work exhibits statistics that can be represented by a quasiprobability in the presence of initial quantum coherence in the energy basis. Here we show that these two concepts of quantum thermodynamics, which apparently appear disconnected, can show a simple relation. Specifically, if a certain work distribution shows negativity asymptotically in the limit of a large number of cells and in a certain time interval, then we surely get a quantum advantage in the charging process. In particular, we prove this for a direct charging protocol performed with a class of charging Hamiltonian operators.
\end{abstract}

\maketitle

\section{Introduction}
Quantum batteries~\cite{Campaioli19,Bhattacharjee21,Campaioli24} are thought of as quantum systems capable of storing energy. This energy can be used to perform useful work, e.g., through a cyclic change in some parameters of the quantum system~\cite{Alicki13}. The maximum extractable work is called ergotropy~\cite{Allahverdyan04} and can be related to some quantum features of the initial state, e.g., quantum correlations~\cite{Francica17,Bernards19,Touil22,Cruz22,Francica22} and quantum coherence~\cite{Francica20,Cakmak20,Niu24}. This amount of work can also be extracted without generating entanglement during the process~\cite{Hovhannisyan13}. Furthermore, all operations that do not generate this useful work have recently been investigated~\cite{Francica25r}.
Regarding battery charge, when the battery is formed by $N$ cells (or copies) the duration time tends to zero as $N$ increases, and thus a quantum advantage is achieved~\cite{Binder15,Campaioli17,Gyhm22}. There are several proposals for the realization of quantum batteries, e.g., by using
many-body interactions in spin systems~\cite{Julia-Farre20}, with a cavity assisted charging~\cite{Ferraro18,Andolina19,Andolina18,Farina19,Zhang19,Andolina192,Crescente20}, in disordered chains~\cite{Rossini19} and fermionic Sachdev-Ye-Kitaev interactions~\cite{Rossini20,Rosa20,Carrega21,Gyhm24}, to name a few. In particular, rescaled Sachdev-Ye-Kitaev interactions~\cite{Francica25,Francica25E}, where the (correlated) disorder is chosen in order to asymptotically saturate the Bhatia-Davis bound~\cite{Bhatia00}, allow to achieve an optimal and extensive quantum advantage.

On the other hand, the work done exhibits statistics that satisfies fluctuation theorems~\cite{Jarzynski97,Crooks99} when it is taken into account with a two-point measurement scheme~\cite{Kurchan00,Tasaki00,Talkner07,Campisi11}, for an incoherent initial state. If there is initial quantum coherence with respect to the energy basis, a no-go theorem~\cite{Perarnau-Llobet17} indicates that the statistics can be represented with some quasiprobability distribution~\cite{Baumer18,Gherardini24,Arvidsson-Shukur24}, i.e., a normalized distribution that also takes negative~\cite{Allahverdyan14,Solinas15} or complex~\cite{Lostaglio24} values. This can be seen as a consequence of the presence of quantum contextuality~\cite{Lostaglio18}. Concerning the form of the quasiprobability, it can be derived by starting from some axioms, e.g., the first law of thermodynamics. For weak conditions, no unique quasiprobability distribution is selected, but we get a class of quasiprobability distributions~\cite{Francica22a,Francica22b}. A description in terms of quasiprobabilities allows one to perform an optimization with a utility function~\cite{Francica24a,Francica24b}, i.e., based on higher moments of the work, and fluctuation theorems can also be  formulated~\cite{Francica24ft}.
However, it is not clear what the role of quasiprobability of work in the charging quantum advantage is, where the duration time plays a key role. Here, we try to answer this question.

To do this, we start by introducing the charging process in Sec.~\ref{sec.process}, which is usually known as the direct charging protocol~\cite{Campaioli24}, and a class of work quasiprobability distributions, identified by a parameter $q$, in Sec.~\ref{sec.quasi}. For our purposes, we focus only on this charging protocol without going beyond the sudden quench setting or generic noise charging models, and we do not take into account alternative definitions of work.
In particular, we give an intuitive simple calculation in Sec.~\ref{subsec.adva-quasip}.  This suggests that there can be some relation between the negativity of some work distribution and the charging quantum advantage. This relation subsists for the work done in a certain time-subinterval of the process and not for any time-subinterval. Furthermore, although there is a quantum advantage, the $q=0,1$ quasiprobability distribution takes positive values, thus in this relation the work distribution is not negative in a representation-independent sense. Hence, we explain how the negativity of a certain work quasiprobability distribution in a certain time interval gives a sufficient condition to achieve quantum advantage in Sec.~\ref{sec.mainres}. In detail, this negativity survives in the limit of a large number of cells, suggesting that the process is strongly quantum. We  further investigate  this result with the help of a hybrid model in Sec.~\ref{sec.ex}. In the end, we summarize and discuss possible applications of our result in Sec.~\ref{sec.conclu}.

\section{Charging process}\label{sec.process}
In order to introduce the charging process, we focus on a quantum battery having $N$ cells and the free Hamiltonian $H_0= \sum_{i=1}^N h_i$, where $h_i$ is the local Hamiltonian of the i-th cell. In our process, the battery is initially prepared in the ground-state of $H_0$, which is $\ket{E_0}=\ket{0}^{\otimes N}$ with $\ket{0}$ such that $h_i\ket{0}=0$ for all $i=1,\ldots,N$. Thus, the lowest energy of the Hamiltonian $H_0$ is $E^0_{min}=E_0=0$. For simplicity, we assume that the ground state is not degenerate. The charging process leads the system to the final state $\ket{\psi(\tau)}=U_{\tau,0}\ket{\psi(0)}$ where $U_{t,0}$ is the unitary operator of time-evolution generated by the time-dependent Hamiltonian $H(t)$ with $t\in[0,\tau]$ such that $H(0)=H(\tau)=H_0$.
For simplicity, let us focus on a time-dependent Hamiltonian such that $H(t)=H_0$ for $t=0,\tau$ and $H(t)=H_1$ for $t\in (0,\tau)$. This situation can be realized by performing two sudden quenches at the initial ($t=0$) and final ($t=\tau$) times. In this case, the unitary operator of time-evolution is $U_{t,0}=e^{-i H_1 t}$ when $t\in(0,\tau)$. We consider a Hamiltonian $H_1$ such that $[H_1,H_0]\neq 0$ in order to achieve a  non-trivial dynamics with a final state $\ket{\psi(\tau)}\propto \ket{E_1}$, where $\ket{E_1}$ is the highest excited state of $H_0$, with energy $E_1=E^0_{max}\sim N$, since we have $\ket{E_1}=\ket{1}^{\otimes N}$ where $\ket{1}$ is the eigenstate of $h_i$ with maximum energy. For simplicity, we assume that the highest excited state is not degenerate.
Thus, at the end of the process, we get a fully charged battery. In order to optimize this charging process, we want to minimize the duration time $\tau$ of the process. To do this, the choice of the charging Hamiltonian $H_1$ is typically constrained by some conditions, e.g., the gap $\Delta E_1 = E_{max}-E_{min}$ between the maximum and minimum eigenvalues of $H_1$, which are $E_{max}$ and $E_{min}$, is not larger than the gap $\Delta E_0 = E^0_{max}$ between the maximum and minimum eigenvalues of $H_0$, which are $E^0_{max}\sim N$ and $E^0_{min}=0$. If no interactions between cells are allowed (parallel charging), the charging Hamiltonian has the form $H_1=\sum_i v_i$, where $v_i$ is a local operator of the i-th cell, and we get a minimum duration time $\tau$ that remains constant and nonzero as $N\to \infty$. On the other hand, if interactions are allowed, e.g., terms like $\sum_{i_1,i_2, \ldots, i_m} \lambda_{i_1i_2 \ldots i_m}v_{i_1}\otimes v_{i_2}\otimes\cdots \otimes v_{i_m}$ may be present in the Hamiltonian $H_1$, we can outperform the parallel charging case, getting a duration time $\tau\to 0$ as $N\to \infty$, and thus we achieve a charging quantum advantage.

In order to give a sufficient scaling argument for a quantum advantage such that $\tau\to 0$ as $N\to \infty$, we note that the average work done in the process in general is $W=\langle H_0 \rangle_\tau-\langle H_0 \rangle_0$.
In general, we define $\langle X \rangle_t =  \Tr{X \rho(t)}$ for any operator $X$ and a density matrix $\rho(t)$. In our case, $t\in[0,\tau]$ and $\rho(t)$ is the time-evolved pure state $\rho(t)=\ket{\psi(t)}\bra{\psi(t)}$ with $\ket{\psi(t)}=U_{t,0}\ket{\psi(0)}$, so that  $\langle X \rangle_t =  \bra{\psi(t)} X \ket{\psi(t)}$.
In particular, we have $\langle H_0 \rangle_t = \sum_i \langle e^{iH_1t}h_ie^{-iH_1t}\rangle_0$ which can be expanded in a  series of powers in $t$ by using the Baker-Campbell-Hausdorff formula, i.e., $\langle H_0 \rangle_t = \sum_i \langle H_1h_iH_1 \rangle_0 t^2 + \cdots$, where we noted that $\ket{\psi(0)}=\ket{E_0}$ and $E_0=E^0_{min}=0$. Since $\ket{\psi(\tau)}\propto \ket{E_1}$ for our process, we get
\begin{equation}\label{eq.cond1}
\frac{E_1}{N}=\frac{1}{N}\sum_i \langle H_1h_iH_1\rangle_0{\tau}^2+\cdots = \sum_k\sum_i \frac{w_{ik}}{N} {\tau}^k \,,
\end{equation}
where we note that the series is convergent. We emphasize that before to perform the limit $N\to\infty$, we have to sum the infinite series, since the two limits are not interchangeable in general. To see how $\tau$ scales with $N$, we can use Eq.~\eqref{eq.cond1}, which of course is satisfied for some duration time $\tau$, i.e., $\tau$ is a solution of Eq.~\eqref{eq.cond1}. Thus, if for some $k$ the coefficient of the k-th term is unbounded, i.e.,
\begin{equation}
\left|\sum_i \frac{w_{ik}}{N}\right|\to \infty
\end{equation}
as $N\to\infty$, then $\tau\to 0$ since the series in Eq.~\eqref{eq.cond1} converges to the sum $E_1/N\sim {\it O}(1)$, so that the series is bounded as $N\to\infty$.
To understand this, we define the function $w(N)$ as the smallest function such that $|\sum_i w_{ik}/N|^{\frac{1}{k}}\leq w(N)$ for all $k$. We define the function $f(x)=\sum_k a_k x^k$ with $a_k = \sum_i w_{ik}/(N (w(N))^{k})$, such that $a_k$ tends to a finite constant (or zero) in the limit $N\to\infty$ for all $k$. Then, $\tau=x^*/w(N)$, where $x^*\sim {\it O}(1)$ is the solution $x=x^*$ of $f(x)=E_1/N$, so that $\tau\to 0$ if $w(N)\to\infty$.
By noting that $|\sum_i w_{i2}|=\langle H_1H_0H_1\rangle_0$, we get that if $\langle H_1H_0H_1\rangle_0/N\to\infty$ as $N\to\infty$ then $\tau\to 0$.
Thus, in this case, there is a quantum advantage.

\section{Quasiprobability distribution}\label{sec.quasi}
Having introduced our charging process, let us focus on the work done in a time interval $[t_1,t_2]$ for a generic out-of-equilibrium process generated by a time-dependent Hamiltonian $H(t)$. We are considering a time interval $[t_1,t_2]$ different from the one previously considered, which was $[0,\tau]$, since we will focus on $t_1\in[0,t_2)$ and $t_2\in(0,\tau]$ in our charging process, later in the paper. In general, the statistics of the work can be represented by some quasiprobability distribution, and for our purposes, we focus on
\begin{eqnarray}
\nonumber && p_q(w) = \sum_{i,j,k} \text{Re}\bra{E_i^{t_1}}\rho(t_1)\ket{E_j^{t_1}}\bra{E_j^{t_1}}U^\dagger_{t_2,t_1}\ket{E_k^{t_2}}\\
&&\times\bra{E_k^{t_2}}U_{t_2,t_1}\ket{E_i^{t_1}}\delta(w-E_k^{t_2}+qE_i^{t_1}+(1-q)E_j^{t_1})\,,
\end{eqnarray}
where $q$ is a real number, $\rho(t_1)$ is the density matrix at time $t=t_1$, and instantaneous eigenvalues $E_i^t$ and eigenstates $\ket{E_i^t}$ are defined by the spectral decomposition $H(t)=\sum_i E_i^t \ket{E_i^t}\bra{E_i^t}$. We recall that the time-evolution operator is obtained as $U_{t_2,t_1}=\mathcal T e^{-i\int_{t_1}^{t_2}H(t)dt}$, where $\mathcal T$ is the time ordering operator. Thus, the moments of the work are defined as $\langle w^n\rangle = \int w^n p_q(w) dw$. In particular, the first three moments read $\langle w \rangle = \langle (H^{(H)}(t_2)-H(t_1))\rangle_{t_1}$, $\langle w^2 \rangle = \langle (H^{(H)}(t_2)-H(t_1))^2\rangle_{t_1}$ and
\begin{eqnarray}
\nonumber\langle w^3 \rangle &=& \langle (H^{(H)}(t_2)-H(t_1))^3\rangle_{t_1} \\
\nonumber&& +\frac{1}{2}\langle[H(t_1)+H^{(H)}(t_2),[H(t_1),H^{(H)}(t_2)]]\rangle_{t_1}\\
&&-3q(1-q)\langle[H(t_1),[H(t_1),H^{(H)}(t_2)]]\rangle_{t_1}\,,
\end{eqnarray}
where the average is calculated with respect to $\rho(t_1)$ as $\langle X \rangle_{t} = \Tr{X\rho(t)}$ for any operator $X$ and we define $H^{(H)}(t_2)=U^\dagger_{t_2,t_1}H(t_2)U_{t_2,t_1}$. We note that only the first two moments do not depend on the parameter $q$, whereas the higher moments, which will play a key role in our analysis, depend on the particular representation chosen, i.e., on the value of $q$. This quasiprobability distribution $p_q(w)$ is the only quasiprobability distribution that satisfies some conditions~\cite{Francica22b} (it reproduces the two-point measurement scheme for incoherent initial states and gives the values for $\langle w\rangle$ and $\langle w^2\rangle$ of above) and helps us to understand what the work statistics in the quantum regime is. When the quasiprobability distribution $p_q(w)$ also takes negative values, we can get some advantage. As noted in Ref.~\cite{Francica24b}, when there is negativity, the moments are not constrained by the Jensen's inequality. Then, given a strictly increasing function $f(w)$ that is concave, i.e., such that $f'(w)>0$ and $f''(w)\leq 0$, if $p_q(w)\geq 0$ for all $w$, from the Jensen's inequality we get the lower bound
\begin{equation}
\langle w \rangle \geq f^{-1}(\langle f(w)\rangle)\,.
\end{equation}
If $f(w)=a w + b$, with $a>0$, this bound is trivially saturated, but for a function $f(w)$ such that $\min_w -f''(w)/f'(w)\geq \varepsilon$ for a given $\varepsilon>0$, the bound cannot be saturated and if the bound is violated then there is negativity. We can refer to $f(w)$ as utility function, as usual. Thus, when $\langle w \rangle$ is negative for some work extraction processes, we can get a larger amount of work extracted $w_{ex}=-\langle w \rangle$ that goes beyond the bound given by the utility $f(w)$, provided that these processes show quasiprobability distributions of work  also taking negative values.
However, differently from the case discussed, it is not clear how the negativity of the quasiprobability can be related to the charging quantum advantage, which involves the duration time $\tau$ and the work moments seem to play no role.

In the end, we note that negativity can be measured with $\mathcal N = \int |p_q(w)|dw$, which is larger than 1, i.e., $\mathcal N>1$,  if and only if $p_q(w)$ takes also negative values~\cite{Francica23}.
We define the characteristic function $\chi_q(u)=\langle e^{iwu}\rangle$, such that the quasiprobability distribution is achieved by performing the Fourier transform $p_q(w)=\frac{1}{2\pi}\int e^{-iwu}\chi_q(u)du$. We get $\chi_q(u)=\frac{1}{2}(X_q(u)+X_{1-q}(u))$, where
\begin{eqnarray}
\nonumber
X_q(u)&=&\text{Tr}\big\{e^{-iuq H(t_1)}U_{t_1,0}\rho(0)U_{t_1,0}^\dagger e^{-iu(1-q)H(t_1)} \\
&& \times e^{iuH^{(H)}(t_2)}\big\}\,,
\end{eqnarray}
where we consider $\rho(t_1)=U_{t_1,0}\rho(0)U_{t_1,0}^\dagger$ and $t_1\in[0,t_2)$, $t_2\in(0,\tau]$.
Thus, negativity can be inferred experimentally by measuring the characteristic function $\chi_q(u)$  through an interferometric measurement scheme (see, e.g., Ref.~\cite{Francica22a}). For instance, we can consider a detector that is a qubit in the initial state $\rho_D(0)$ with Hamiltonian $H_D=\omega_D \ket{e}\bra{e}$. We consider $H_I=-\delta_e \ket{e}\bra{e}+\delta_g \ket{g}\bra{g}$ and $H'_I=-\delta'_e \ket{e}\bra{e}+\delta'_g \ket{g}\bra{g}$, where $\ket{g}$ is the ground-state of the qubit and $\ket{e}$ is the excited state. Thus, if the total Hamiltonian is $H_{tot}(t)= H(t)+H_D + \delta(t-t_1) H(t_1)\otimes H_I+\delta(t-t_2)H(t_2)\otimes H'_I$, the coherence of the qubit at the final time $\tau$ reads
\begin{eqnarray}\nonumber
\bra{e}\rho_D(\tau)\ket{g}&=&\bra{e}\rho_D(0)\ket{g}e^{-i\omega_D \tau}\text{Tr}\big\{e^{i\delta_e H(t_1)}U_{t_1,0}\rho(0)\\
&& \times U_{t_1,0}^\dagger e^{i\delta_g H(t_1)} e^{i(\delta'_e+\delta'_g)H^{(H)}(t_2)}\big\}\,,
\end{eqnarray}
from which we can determine $X_q(u)$.

\subsection{Quantum advantage and work quasiprobability}\label{subsec.adva-quasip}
To explain the intuition behind our main result, we focus on the charging process of Sec.~\ref{sec.process} with $H_1=\lambda (\ket{E_1}\bra{E_0}+\ket{E_0}\bra{E_1})$. In this case, the time-evolved state is $\ket{\psi(t)}=\cos(\omega t/2)\ket{E_0}-i\sin(\omega t/2)\ket{E_1}$, where $\omega=2\lambda$. The duration time is $\tau=\pi/\omega$, so we get the final state $\ket{\psi(\tau)}=-i\ket{E_1}$ and the battery is fully charged. If $\lambda\to \infty$ as $N\to\infty$, then $\tau\to 0$ and there is a quantum advantage. Let us examine the statistics of the work done in a time interval $[t_1,t_2]$. For $t_1=0$ there is no coherence with respect to the energy basis at the initial time $t_1$, then the work distribution is a probability distribution. For $t_1\in (0,t_2)$ and $t_2<\tau$ we get $H(t_1)=H(t_2)=H_1$ and since the time-evolution is generated by $H_1$, we get a trivial situation with a probability distribution for the work. In contrast, the case $t_1\in (0,\tau)$ and $t_2=\tau$ appears to be more interesting. In particular for $t_1=\tau/2$ and $t_2=\tau$ it is easy to see that the work quasiprobability distribution $p_q(w)$ reads
\begin{equation}
p_q(w)=\frac{1}{4}\sum_{i=\pm,k=0,1}\delta(w-E_k+E_i)+\delta p_q(w)\,,
\end{equation}
where $E_\pm$  are the only eigenvalues of $H_1$ different from zero and $\ket{E_\pm}$ are the corresponding eigenstates.
In detail, we define
\begin{eqnarray}
\nonumber\delta p_q(w)&=&\frac{1}{4}\sum_{k=0,1}(-1)^{k+1}(\delta(w-E_k+q E_++(1-q)E_-)\\
&&+\delta(w-E_k+q E_-+(1-q)E_+))\,,
\end{eqnarray}
which, of course, also takes negative values.
Then, the quasiprobability distribution $p_q(w)$ is a probability distribution for $q=0,1$, explicitly $p_0(w)=\frac{1}{2}\delta(w-E_1+E_-)+\frac{1}{2}\delta(w-E_1+E_+)$. In contrast, for $q\neq 0,1$, $p_q(w)$ also takes negative values due to the term $\delta p_q(w)$. Thus, the charging quantum advantage appears to be related to the negativity of $p_q(w)$. However, negativity does not depend on the support of $p_q(w)$, while $\tau\sim 1/\lambda$ depends on the scaling of $\lambda$ with $N$. Then, there is no relation between the magnitude of negativity and the degree of quantum advantage.
In particular, the negativity achieved is not related to a violation of a Leggett-Garg inequality, which is always satisfied (see Appendix~\ref{sec.leggett-garg}).

\section{Main result}\label{sec.mainres}

In the following, we consider the charging process of Sec.~\ref{sec.process} with $t_1\in (0,\tau)$ and $t_2=\tau$. We get the characteristic function $\chi_q(u)=\frac{1}{2}(X_q(u)+X_{1-q}(u))$ with
\begin{equation}\label{eq.XXX}
X_q(u)=\langle e^{-iu(1-q)H_1} e^{iuH_0}e^{-iuqH_1} \rangle_\tau\,,
\end{equation}
which does not depend on the particular value of $t_1$.
Furthermore, the cumulant generating function is defined as $G_q(u)=\ln \chi_q(u)$ and the n-th cumulant is $\kappa_n=(-i)^n G_q^{(n)}(0)$.
For example, for $q=1/2$ we get $\chi_{1/2}(u)=X_{1/2}(u)$. In particular, since $X_{1/2}(-u)=X^*_{1/2}(u)$, the Fourier transform of $X_{1/2}(u)$ takes real values, and thus $p_{1/2}(w)$ is a real valued distribution.
In contrast, for $q=0$, we get $\chi_0(u)=X_0(u)=X_1(u)=e^{iuE^0_{max}}\langle e^{-iuH_1}\rangle_\tau$.
This means that the quasiprobability distribution for $q=0$ is equal to the probability distribution of an observable.
Then, the cumulants $\kappa_n$ of the work for $q=0$ that are higher than the first cumulant $\kappa_1$ are equal to the cumulants $\kappa'_n$ of the Hamiltonian $-H_1$ coming from the generating function $G'(u)=\ln \langle e^{-iuH_1}\rangle_\tau$, which are $\kappa'_n=(-i)^n G'^{(n)}(0)$. In particular, the first cumulants are $\kappa'_1=-\langle H_1\rangle_\tau$, $\kappa'_2 = \langle H_1^2\rangle_\tau-\langle H_1\rangle_\tau^2$ and $\kappa'_3=-\langle H_1^3\rangle_\tau+3\langle H_1^2\rangle_\tau\langle H_1\rangle_\tau-2 \langle H_1\rangle_\tau^3$.
Since the first two moments do not depend on $q$, we see that the first two cumulants $\kappa_n$ are the same for any value of $q$. Thus, from the expression derived for $\chi_0(u)$, we get $\kappa_1=E^0_{max}-\langle H_1\rangle_\tau$ and $\kappa_2=\kappa_2'$. In particular, the first cumulant $\kappa_1$ is the average work done in the interval $[t_1,t_2]$. In contrast, the third cumulant reads
\begin{equation}\label{eq.thirdcumulant}
\kappa_3= \kappa'_3 -6 q(1-q)\langle H_1 \tilde H_0 H_1 \rangle_\tau\,,
\end{equation}
where we defined the operation of inversion of the spectrum of a hermitian operator $X$ as $\tilde X = X_{max}-X$, where $X_{max}$ is the maximum eigenvalue of $X$, so that $\tilde H_0 = E^0_{max}-H_0$. If the eigenvalues of $\tilde H_0$ are equal to the eigenvalues of $H_0$, this operation can be realized through some unitary operator $U_I$ that squares to the identity, i.e., $\tilde H_0=U_I H_0 U_I$ with $U_I^2=I$.

In general, we define $g_q(u)$ such that
\begin{equation}
X_q(u)=e^{iuE^0_{max}}e^{N g_q(u)}\,.
\end{equation}
The variance of the work distribution depends on $g_q''(0)$, and is nonzero if and only if $g''_q(0)\neq 0$.
We note that $g'_q(0)$ is imaginary, $g''_q(0)$ can be complex and $g''_{1/2}(0)$ is real,  $g''_{1/2}(0)=-\kappa_2/N$.
In particular, for $q=1/2$ the cumulant generating function is $G_{1/2}(u)=iuE^0_{max}+N g_{1/2}(u)$ and $\kappa_n=(-i)^n N g^{(n)}_{1/2}(0)$ for $n\geq2$.
We define the quasiprobability distribution
\begin{equation}
\hat p_q(w)=\sqrt{N} p_q(\sqrt{N}w)\,.
\end{equation}
We aim to analyze the properties of the asymptotic formula of this distribution as $N\to\infty$.
Its characteristic function is $\hat \chi_q(u)=\int e^{iuw}\hat p_q(w)dw = \chi_q(u/\sqrt{N})$.
Thus $\hat p_q(w)$ can be represented with a binned histogram, overcoming the delta peaks contained in the distribution. Since the length of the support of $\hat p_q(w)$ scales as $\sqrt{N}$, we can consider a number $M\sim \sqrt{N}$ of bins in the histogram. In particular, the histogram $\hat p^h_q(w)$ can be obtained from the characteristic function as explained in Appendix~\ref{sec.histogram}.
We proceed by giving a formal cumulant argument similar to the proof of the central limit theorem. Basically, as we will see, under certain circumstances the main non-zero contribution for $q=1/2$ comes from a neighbourhood of the mean value, as $N\to\infty$.
If $|g^{(n)}_{q}(0)|< \infty$ as $N\to \infty$ for all $n$, then we get
\begin{equation}\label{eq.Xq}
X_q\left(\frac{u}{N^{\frac{1}{2}}}\right)\sim e^{iuN^{-\frac{1}{2}}E^0_{max}+N^{\frac{1}{2}} g_q'(0)u+\frac{g_q''(0)}{2}u^2+{\it O}(N^{-\frac{1}{2}})}\,,
\end{equation}
from which we note that the skewness is ${\it O}(N^{-\frac{1}{2}})$ and tends to zero as $N\to\infty$.
By neglecting ${\it O}(N^{-\frac{1}{2}})$  terms, from Eq.~\eqref{eq.Xq} we get
\begin{equation}\label{eq.centra}
\hat p_q(w)\sim \frac{1}{\sqrt{2\pi}}\text{Re}\left(\frac{\exp\left(\frac{(\sqrt{N}w-E^0_{max}+iNg_q'(0))^2}{2Ng_q''(0)}\right)}{\sqrt{-g_q''(0)}}\right)\,,
\end{equation}
which for $q=1/2$ reduces to the Gaussian
\begin{equation}\label{eq.centraq12}
\hat p_{1/2}(w)\sim \frac{1}{\sqrt{2\pi}}\frac{\exp\left(\frac{(\sqrt{N}w-E^0_{max}+iNg_{1/2}'(0))^2}{2Ng_{1/2}''(0)}\right)}{\sqrt{-g_{1/2}''(0)}}\,.
\end{equation}
The same result is derived by noting that the cumulants corresponding to the quasiprobability distribution $\hat p_q(w)$ are $\hat \kappa_n = \kappa_n/N^{\frac{n}{2}}$. If $|g^{(n)}_{q}(0)|< \infty$ as $N\to \infty$ for all $n$, then $\hat \kappa_n \lesssim {\it O}(N^{-\frac{1}{2}})$ for $n\geq 3$. Then, $\hat \kappa_n \to 0$ as $N\to\infty$ for $n\geq 3$ and the function $g_q(u/\sqrt{N})$, appearing in $X_q(u/\sqrt{N})$, can be replaced by a quadratic polynomial in $u$ in the limit $N\to \infty$, and thus we achieve Eq.~\eqref{eq.centra}.
We note that the rate at which $\hat p_{1/2}(w)$ approaches the asymptotic Gaussian in Eq.~\eqref{eq.centraq12} is dictated by the largest error term. Then, the difference between  $\hat p_{1/2}(w)$ and the asymptotic Gaussian in Eq.~\eqref{eq.centraq12} tends to zero no slower than $1/\sqrt{N}$ as $N\to\infty$ for all $w$ in the support. This matches the well-known Berry-Esseen theorem. Then, for values of $w$ where the Gaussian tends to zero as $N\to\infty$, even if $\hat p_{1/2}(w)$ can take negative values, $\hat p_{1/2}(w)$ tends to zero no slower than $1/\sqrt{N}$ as $N\to\infty$.
Thus, $\hat p_q(w)$ is positive for $q=1/2$ as $N\to\infty$, i.e., $\hat p_{1/2}(w)\geq 0$ in the limit $N\to\infty$ for all $w$ in the support, and also takes negative values for $q\neq 1/2$ if $g_q''(0)$ is not real.
Therefore, for $q= 1/2$ if $\hat p_{1/2}(w)<0 $ for some $w$  as $N\to \infty$, then  $g^{(n)}_{1/2}(0)$ is unbounded for some $n$, i.e.,  $|g^{(n)}_{1/2}(0)|\to \infty$ as $N\to \infty$ for some $n$.
Similarly, if $g''_q(0)$ is real, from Eq.~\eqref{eq.centra} we get $\hat p_q(w)\sim \hat p_{1/2}(w)$. Thus, in this case, if $\hat p_{q}(w)<0 $ for some $w$  as $N\to \infty$, then  $g^{(n)}_{q}(0)$ is unbounded for some $n$.

In order to proceed with the analysis, we define the class $\mathcal C$ of systems that are lattice models with local dimension $d$ and total number of sites $N$. The Hamiltonian reads $H_1=\sum_X v_X$, where $v_X=0$ if $|X|>r$ for some $r$. We assume that there are $k$ terms in the sum over $X$, with $k\leq N$. Moreover, we assume that $||H_1||\sim N$ and $||v_X||\sim \lambda$. In the end, we assume that the order of the interactions is $m\sim r^\nu$, where $\nu$ is the dimension of the lattice, i.e., we assume that for $|X|\leq r$
\begin{equation}
v_X = \lambda_{i_1i_2 \ldots i_m}v^1_{i_1}\otimes v^2_{i_2}\otimes\cdots \otimes v^m_{i_m}\,,
\end{equation}
with $m\sim r^\nu$, where $X=\{i_1,i_2,\ldots,i_m\}$ and $\lambda_{i_1i_2 \ldots i_m}\sim \lambda$, and 
$[v^k_{i_k},H_0]\neq 0$ for 
all local operators $v^k_{i_k}$ on the site $i_k$. 
For simplicity, we also assume that $\langle v^k_{i_k}\rangle_0=\langle v^k_{i_k} \rangle_\tau=0$ for all $k$ and $i_k$. All these assumptions define the class $\mathcal C$.

We start to focus on systems with charging Hamiltonian $H_1\in \mathcal C$ such that (i) $\kappa'_3/N$ is bounded as $N\to\infty$. Then, from Eq.~\eqref{eq.thirdcumulant}, we get the following sufficient condition to get a quantum advantage.
\begin{lemma}\label{lemma.suffcond2}
For $H_1\in \mathcal C$, when $|\kappa'_3|/N<\infty$ as $N\to\infty$, if $|\kappa_3|/N\to\infty$ as $N\to\infty$ then $\tau\to 0$.
\end{lemma}
\begin{proof}
If $|\kappa_3|/N\to\infty$ as $N\to\infty$, then, from Eq.~\eqref{eq.thirdcumulant}, we get $\langle H_1 \tilde H_0 H_1 \rangle_\tau/N\to \infty$ since $\kappa'_3/N$ is bounded. In particular, this can happen only for $q\neq0,1$.
In the end, we note that the term $\langle H_1 \tilde H_0 H_1 \rangle_\tau$ in $\kappa_3$, appearing for $q\neq0,1$, scales as $\langle H_1 H_0 H_1 \rangle_0$.   To show this, we consider $H_1\in \mathcal C$. The charging Hamiltonian reads $H_1=\sum_X v_X$. By noting that $h_i \ket{\psi(0)}=0$, we get
\begin{equation}
\langle H_1 H_0 H_1\rangle_0 =  \sum_{i,X,X'|i\in X\cap X'} \langle v_X h_i v_{X'}\rangle_0 \,.
\end{equation}
Furthermore, since  $\langle v^k_{i_k}\rangle_0=0$ for all $k$ and $i_k$, we get
\begin{equation}\label{eq.proof---}
\langle H_1 H_0 H_1\rangle_0 =  \sum_{i,X|i\in X} \langle v_X h_i v_{X}\rangle_0 \,.
\end{equation}
Similarly, by noting that $\tilde h_i \ket{\psi(\tau)}=0$, we get
\begin{equation}
\langle H_1 \tilde H_0 H_1\rangle_\tau =  \sum_{i,X,X'|i\in X\cap X'} \langle v_X \tilde h_i v_{X'}\rangle_\tau \,,
\end{equation}
and, since  $\langle v^k_{i_k}\rangle_\tau=0$ for all $k$ and $i_k$, we get
\begin{equation}\label{eq.proof}
\langle H_1 \tilde H_0 H_1\rangle_\tau =  \sum_{i,X|i\in X} \langle v_X \tilde h_i v_{X}\rangle_\tau \,.
\end{equation}
Then, since both sums in Eqs.~\eqref{eq.proof---} and~\eqref{eq.proof} involve the same number of positive numbers, which scale as $\lambda^2$, we get $\langle H_1 \tilde H_0 H_1 \rangle_\tau\sim\langle H_1 H_0 H_1 \rangle_0$.
In particular, we get the equality $\langle H_1 \tilde H_0 H_1 \rangle_\tau=\langle H_1  H_0 H_1 \rangle_0$ if the inversion symmetry $U_I H_1 U_I = H_1$ is satisfied. Then, from the sufficient scaling argument of Sec.~\ref{sec.process}, 
 we get $\tau\to 0$.
\end{proof}
Furthermore, when condition (i) is satisfied, for any Hamiltonian $H_1\in \mathcal C$, for $q\neq 0,1$, if $|g^{(n)}_{q}(0)|\to\infty$ as $N\to \infty$ for some $n$, then $|g'''_{q}(0)|\to\infty$.
\begin{proof}
We have to show that $|g'''_{q}(0)|<\infty$ for $N\to \infty$ implies $|g^{(n)}_{q}(0)|<\infty$ for $N\to \infty$ for all $n$. To show this, we consider  $H_1\in \mathcal C$. We get
\begin{equation}
N\sim ||H_1|| \leq \sum_X ||v_X||\sim \lambda k\,.
\end{equation}
Thus, $\lambda k \geq c N$ for some positive constant $c$.
From Eq.~\eqref{eq.proof}, if the order of the interactions is $m\sim r^\nu$, we get
\begin{equation}\label{eq.scal3}
\langle H_1 \tilde H_0 H_1\rangle_\tau = \sum_{i,X|i\in X} \langle v_X \tilde h_i v_{X}\rangle_\tau \sim \lambda^2 r^\nu k\,.
\end{equation}
Then, if condition (i) holds, since $\kappa_3=iNg'''_q(0)$, we get $|g'''_q(0)|\geq c' \lambda^2 r^\nu k /N$ for $q\neq 0,1$, and thus $\lambda^2 r^\nu k/N<\infty$. Since $\lambda k \geq c N$, we get $|g'''_q(0)|\geq c'' r^\nu N/k$, then $r^\nu N/k <\infty$. If $N/k\to \infty$, $|g'''_q(0)|\to\infty$, then $1\leq N/k<\infty$ and thus $r<\infty$.  Then $v_X$ is local and the state
\begin{equation}
\ket{\tilde \psi_q(u)}= e^{-iu\tilde H_0} e^{-iuqH_1}\ket{1}^{\otimes N}
\end{equation}
is short range entangled~\cite{Chen10}, thus can be expressed as a tensor product state with finite bond dimension. For $\nu=1$, we get the matrix product state
\begin{equation}
\ket{\tilde \psi_q(u)}=\sum_{i_1,\ldots,i_N} A_{i_1}\cdots A_{i_N} \ket{i_1\ldots i_N}\,,
\end{equation}
where $A_{i_k}$ are $D_{k-1}\times D_{k}$ matrices with $D_k<\infty$ as $N\to \infty$.
Similarly, the state
\begin{equation}
\ket{\psi_{1-q}(-u)}= e^{iu(1-q)H_1}\ket{1}^{\otimes N}
\end{equation}
can be expressed as
\begin{equation}
\ket{\psi_{1-q}(-u)}=\sum_{i_1,\ldots,i_N} B_{i_1}\cdots B_{i_N} \ket{i_1\ldots i_N}\,,
\end{equation}
where $B_{i_k}$ are $D_{k-1}\times D_{k}$ matrices with $D_k<\infty$ as $N\to \infty$.
The overlap between these two states reads
\begin{equation}\label{eq.mpsoverlap}
\braket{\psi_{1-q}(-u)}{\tilde\psi_q(u)} = \sum_{i_1,\ldots,i_N} B_{i_N}^\dagger\cdots B_{i_1}^\dagger A_{i_1}\cdots A_{i_N}\,.
\end{equation}
By performing a singular value decomposition, the matrix $A_{i_k}$ can be expressed as $A_{i_k}=U_{i_k} \Lambda_{i_k} V_{i_k}^\dagger$, where $U_{i_k}$ and $V_{i_k}$ are unitary matrices and $\Lambda_{i_k}=\sum_j a_{i_k,j}\ket{e_j}\bra{e_j}$ with $a_{i_k,j}\geq 0$, thus $A_{i_k}=\sum_j a_{i_k,j} \ket{a'_{i_k,j}}\bra{a_{i_k,j}} $ with $\ket{a'_{i_k,j}}=U_{i_k}\ket{e_j}$ and $\ket{a_{i_k,j}}=V_{i_k}\ket{e_j}$. In particular, for $u=0$ we get $A_{i_k}=\delta_{i_k 1}\ket{e_1}\bra{e_1}$. Similarly, $B_{i_k}=\sum_j b_{i_k,j} \ket{b'_{i_k,j}}\bra{b_{i_k,j}} $, then  from Eq.~\eqref{eq.mpsoverlap} we get
\begin{equation}
\braket{\psi_{1-q}(-u)}{\tilde\psi_q(u)}\sim \prod_{k=2}^N a_k b_k \braket{a'_k}{a_{k-1}}\braket{b_k}{b'_{k-1}}\,,
\end{equation}
where $a_k$ and $b_k$ are certain optimal elements in the sets $\{a_{i_k,j}\}$ and $\{b_{i_k,j}\}$, e.g., we can have (but not necessarily) $a_k = \max_{i_k,j} a_{i_k,j}$ and $b_k=\max_{i_k,j}b_{i_k,j}$, such that $a_k=b_k=1$ and $\ket{a_k}=\ket{a_k}=\ket{b_k}=\ket{b'_k}=\ket{e_1}$ for $u=0$. Then, by noting that
\begin{equation}
g_q(u) = \frac{1}{N} \ln \braket{\psi_{1-q}(-u)}{\tilde\psi_q(u)}
\end{equation}
we get
$g_q(u)\sim \frac{1}{N}\sum_k g_{q,k}(u)$
where $g_{q,k}(0)=0$, thus $g^{(n)}_q(0)\sim {\it O}(1)$ for all $n$. For $\nu>1$ the proof is completely similar.
\end{proof}
In contrast, when condition (i) is not satisfied, i.e., $\kappa'_3/N$ is unbounded, for $H_1\in\mathcal C$ we get $\tau\to 0$, and there is a quantum advantage. In general, for $H_1\in\mathcal C$ we have the following lemma.
\begin{lemma}\label{lemma.k3unbounded}
For $H_1\in \mathcal C$, if $|\kappa'_n|/N\to\infty$ as $N\to\infty$ for some $n\geq2$, then $\tau\to 0$.
\end{lemma}
\begin{proof}
To show this, we consider $H_1\in \mathcal C$. Given $X$ such that $|X|\leq r$, the number of $X'$ such that $|X'|\leq r$ and $X\cap X'\neq \emptyset$ is $r'<ar^\nu$ for some positive constant $a$, then
\begin{eqnarray}
\kappa'_2 &=& \sum_{X,X'|X\cap X'\neq \emptyset} \langle v_X v_{X'}\rangle_\tau-\langle v_X \rangle_\tau \langle v_{X'}\rangle_\tau\\
&\lesssim& k r' \lambda^2\,,\\
\nonumber \kappa'_3 &=& \sum_{X,X',X''|X\cap X'\cap X''\neq \emptyset} -\langle v_X v_{X'}v_{X''}\rangle_\tau\\
&&+3\langle v_X v_{X'}\rangle_\tau\langle v_{X''}\rangle_\tau-2 \langle v_X \rangle_\tau \langle v_{X'} \rangle_\tau\langle v_{X''} \rangle_\tau\\
|\kappa'_3|&\lesssim& k r'^2 \lambda^3\,,
\end{eqnarray}
and so on.
If $||H_1||\sim N$, then $\lambda \gtrsim N/k$, thus for $n\geq 2$, $|\kappa'_n|/N \lesssim (r'\lambda)^{n-1}k\lambda/N$ with $k\lambda/N\gtrsim 1$. Hence, if $|\kappa'_n|/N\to\infty$ as $N\to\infty$ for some $n\geq2$, then we can have (a) $k\lambda/N<\infty$ and $r'\lambda\to \infty$, (b) $k\lambda/N\to\infty$ and $r'\lambda>0$, or (c) $k\lambda/N\to\infty$ and $r'\lambda\to 0$. Similarly to Eq.~\eqref{eq.scal3}, if the order of the interactions is $m\sim r^\nu$, we have
\begin{equation}
\langle H_1 H_0 H_1\rangle_0 = \sum_{i,X|i\in X} \langle v_X h_i v_{X}\rangle_0 \sim \lambda^2 r^\nu k\,,
\end{equation}
thus $\langle H_1 H_0 H_1\rangle_0/N\sim r^\nu\lambda^2 k/N > c' r' \lambda^2 k/N$, so that if $|\kappa'_n|/N\to\infty$ for some $n\geq2$, then  $\langle H_1 H_0 H_1\rangle_0/N\to \infty$ (if (a) or (b) holds the proof is straightforward, if (c) holds it is enough to note that $r'\lambda\geq (r'\lambda)^{n-1}$) and thus from the sufficient scaling argument of Sec.~\ref{sec.process} 
 we get $\tau \to 0$.
\end{proof}
Thus, we deduce our main result, which is the following theorem.
\begin{theorem}\label{theorem}
For the charging process of Sec.~\ref{sec.process}, for $t_1\in (0,\tau)$ and $t_2=\tau$, and for any charging Hamiltonian $H_1\in\mathcal C$, if for $q=1/2$ the asymptotic rescaled quasiprobability distribution of the work $w$ done in the time interval $[t_1,t_2]$ also takes negative values, i.e., $\hat p_{1/2}(w)<0$ for some $w$ as $N\to\infty$, then $\tau \to 0$ as $N\to \infty$, and thus there is a quantum advantage in the charging process.  Furthermore, for $q\neq 1/2$, if $g''_q(0)$ is real and $\hat p_{q}(w)<0$ for some $w$ as $N\to\infty$, then $\tau \to 0$ as $N\to \infty$.
\end{theorem}
\begin{proof}
For $H_1\in\mathcal C$, if condition (i) is not satisfied from lemma~\ref{lemma.k3unbounded} we get $\tau\to 0$. In contrast, if condition (i) is satisfied, then, 
if $\hat p_{1/2}(w)<0$ for some $w$ as $N\to\infty$, $|g'''_{1/2}(0)|\to\infty$ as $N\to \infty$. Thus, $|\kappa_3|/N\to\infty$ and from lemma~\ref{lemma.suffcond2} we get $\tau\to 0$.
For $q\neq 1/2$, if $g''_q(0)$ is real, the proof is analogous.
\end{proof}
Basically, the theorem guarantees that, in some circumstances, if $\hat p_{1/2}$ shows negativity that does not vanish in the limit of a large number of cells $N\to\infty$, then we get a quantum advantage with $\tau\to 0$ in this limit.
We observe that the quasiprobability distribution for $q=1/2$ plays a key role since $g''_q(0)$ is surely real for $q=1/2$. For $q\neq 1/2$ negativity survives in the limit $N\to \infty$ if $g''_{q}(0)$ is complex although $\tau$ does not tend to zero as $N\to\infty$.
Of course, we note that the converse of theorem~\ref{theorem}  does not hold, since, in principle, we can achieve $\kappa_3/N$ unbounded for a nonnegative probability distribution $\hat p_{1/2}(w)$, thus negativity gives just a sufficient condition to achieve $\tau\to 0$.  We show this with the help of a specific example in the next section.

To explain the origin of this quantum advantage, we note that the presence of negativity in $\hat p_{1/2}(w)$ implies a non-local charging Hamiltonian $H_1$. Thus, the quantum advantage in this case is due to the non-local structure of $H_1$. To understand this, we consider a lattice model with local dimension $d$ and total number of sites $N$. If $H_1$ is a r-local Hamiltonian, then $H_1=\sum_X v_X$ with $v_X=0$ if $|X|>r$. If $r<\infty$ as $N\to\infty$, we get
\begin{equation}
|\langle H_1 \tilde H_0 H_1\rangle_\tau|\leq  \sum_i  \sum_{X|i\in X}\sum_{X'|i\in X'} |\langle v_X \tilde h_i v_{X'}\rangle_\tau| \leq c N
\end{equation}
for some constant $c$. Then, since $\kappa'_3/N$ is bounded, $\kappa_3/N$ is bounded. Similarly, $g^{(n)}_{1/2}(0)$ is bounded for all $n$. Then,  
if $\hat p_{1/2}(w)<0$ for some $w$ as $N\to\infty$, then $g^{(n)}_{1/2}(0)$ is unbounded for some $n$, which implies that $r\to\infty$ as $N\to\infty$, i.e., the charging Hamiltonian $H_1$ is non-local.
This charging Hamiltonian generates strong quantum coherence during time evolution; otherwise the quasiprobability distribution tends to be positive (in this case it reproduces the two-point measurement scheme probability distribution). Therefore, physical causes of this quantum advantage could be sought in these elements.

\section{Example} \label{sec.ex}

We focus on a local Hilbert space (of a single cell) having dimension $d=2$. Thus, we consider $h_i=\epsilon_0(\sigma^z_i+1)/2$, where $\sigma^x_i$, $\sigma^y_i$ and $\sigma^z_i$ are the local Pauli matrices of the i-th cell. We consider the charging Hamiltonian of Ref.~\cite{Julia-Farre20}
\begin{equation}\label{eq.H1no}
H_1 = \lambda \sum_{j=0}^{k-1}\otimes_{i=1}^r\sigma^x_{rj+i}+\lambda \otimes_{i=1}^s\sigma^x_{rk+i}\,,
\end{equation}
where the second term is absent for $s=0$, with $N=kr+s$, where $k=[N/r]$, $0\leq s < r$, and $[x]$ is the integer part of $x$. For $s=0$ the time-evolved state is
\begin{equation}
\ket{\psi(t)}=e^{-iH_1 t} \ket{0}^{\otimes N}=\otimes_{j=0}^{k-1}\ket{\psi_j(t)}\,,
\end{equation}
where
\begin{equation}
\ket{\psi_j(t)}=\sqrt{1-p(t)}\otimes_{i=1}^r\ket{0}_{rj+i}-i\sqrt{p(t)}\otimes_{i=1}^r\ket{1}_{rj+i}\,,
\end{equation}
with $p(t)=\sin^2(\lambda t)$. The local states $\ket{0}_i$ and $\ket{1}_i$ are the eigenstates of $\sigma^z_i$ with eigenvalues $-1$ and $1$, respectively. For $s>0$ the time-evolved state is
\begin{equation}
\ket{\psi(t)}=\otimes_{j=0}^{k}\ket{\psi_j(t)}\,,
\end{equation}
where
\begin{equation}
\ket{\psi_k(t)}=\sqrt{1-p(t)}\otimes_{i=1}^s\ket{0}_{rk+i}-i\sqrt{p(t)}\otimes_{i=1}^s\ket{1}_{rk+i}\,.
\end{equation}
Thus, the final time is
\begin{equation}\label{eq.finaltime}
\tau=\frac{\pi}{2\lambda}\,.
\end{equation}
For $s=0$, the maximum and minimum eigenvalues of $H_1$ are $E_{max}=\lambda k$ and $E_{min}=-\lambda k$, then their difference is $\Delta E_1=2\lambda k$. In contrast, for $s>0$ we get $\Delta E_1=2\lambda (k+1)$ due to the partial block. Requiring $\Delta E_1 \sim N$ as $N\to\infty$, we get $\lambda\sim r$, and thus $\tau\sim 1/r$. Then, if $r\to \infty$ as $N\to \infty$ we get $\tau\to 0$ and thus a quantum advantage. We note that $H_1\in \mathcal C$, then theorem~\ref{theorem} applies.
For this model, for $t_1\in (0,\tau)$ and $t_2=\tau$, from Eq.~\eqref{eq.XXX} we get in the limit $N\to \infty$
\begin{equation}
X_q(u)\sim e^{iuE^0_{max}}\left(c_q(u)-s_q(u)e^{-iur\epsilon_0}\right)^k\,,
\end{equation}
where $E^0_{max}=N\epsilon_0$, $c_q(u)=\cos(\lambda q u )\cos(\lambda(1-q)u)$ and $s_q(u)=\sin(\lambda q u )\sin(\lambda(1-q)u)$. Thus, $g_q(u)$ reads
\begin{equation}
g_q(u)\sim\frac{1}{r}\ln(c_q(u)-s_q(u)e^{-iur\epsilon_0})\,,
\end{equation}
from which we get $g'_q(0)=0$, $g''_q(0)=-\lambda^2/r\sim r$, $g'''_q(0)=6iq(1-q)\epsilon_0\lambda^2\sim r^2$. Furthermore, $U_I=\otimes^N_{i=1}\sigma^x_i$ and $U_I H_1 U_I=H_I$. Then condition (i) is satisfied. In order to calculate $\hat p_q(w)$, we consider $g_q(u/\sqrt{N})$ as $N\to\infty$. If $r/\sqrt{N}\to 0 $ as $N\to\infty$, then $g_q(u/\sqrt{N})\sim - \lambda^2 u^2/(2rN)$ so that
\begin{equation}
\hat p_{q}(w)\sim \frac{\sqrt{N}}{\sqrt{2\pi}\sigma}e^{-\frac{(\sqrt{N}w-E^0_{max})^2}{2\sigma^2}}\,,
\end{equation}
with $\sigma^2=N\lambda^2/r$. Then, $\hat p_{q}(w)$ is a Gaussian, $\hat p_{q}(w)>0$ although $\tau\to 0$ if $r\to\infty$ as $N\to \infty$.
In contrast, if $r/\sqrt{N}\to \infty $ as $N\to\infty$, $\hat p_{1/2}(w)$ also takes negative values as shown in Fig.~\ref{fig:plot1}.
\begin{figure}
[t!]
\centering
\includegraphics[width=0.9\columnwidth]{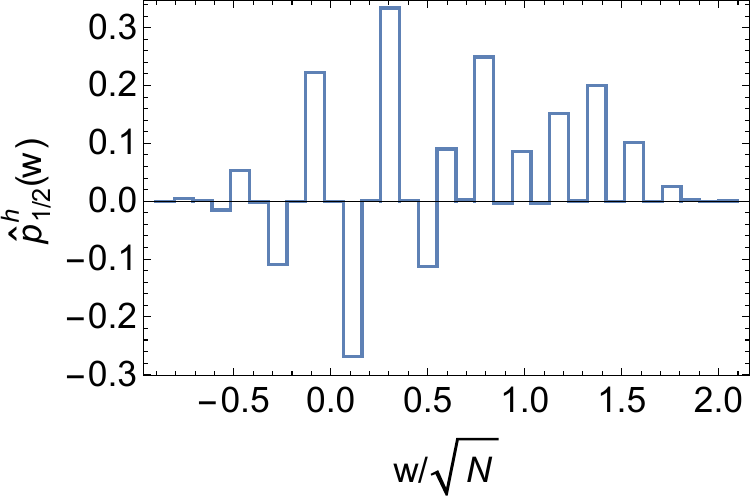}
\caption{The histogram $\hat p^h_{1/2}(w)$ of the quasiprobability distribution $\hat p_{1/2}(w)$. We put $N=1000$, $r=[N^{0.75}]$, $\lambda=r$ and $\epsilon_0=1$. We divide the support in a number $M=[\sqrt{N}]$ of bins, and the normalization of the plotted quasiprobability histogram is such that the sum of the heights of the bins is one.
}
\label{fig:plot1}
\end{figure}

Finally, it is interesting to investigate whether the negativity is robust to perturbations. To do this, we consider a perturbation $H_p=\alpha\sum_{i=1}^N \sigma^x_i$, so that the charging Hamiltonian reads $H'_1=H_1+H_p$. For $r\to\infty$ as $N\to\infty$, at the final time in Eq.~\eqref{eq.finaltime} the final state is $\ket{\psi(\tau)}\sim \ket{1}^{\otimes N}+{\it O}(1/r)$, thus the battery tends to be fully charged as $N\to\infty$. Concerning the quasiprobability distribution $p_{1/2}(w)$, we get
\begin{eqnarray}\nonumber
g_{1/2}(u)&\sim&\frac{1}{r}\ln(\alpha_{1/2}^r(u)c_{1/2}(u)-\beta_{1/2}(u)\\
&&-(\alpha_{1/2}^r(u))^*s_{1/2}(u)e^{-iur\epsilon_0})\,,
\end{eqnarray}
where
\begin{eqnarray}
\alpha_{1/2}(u)&=&\cos^2(\alpha u/2)-e^{-iu\epsilon_0}\sin^2(\alpha u/2)\,,\\
\beta_{1/2}(u)&=&i \sin(\lambda u)(-i\sin(\alpha u)(1+e^{-iu\epsilon_0})/2)^r\,.
\end{eqnarray}
As shown in Fig.~\ref{fig:plot2}, negativity persists for this commuting perturbation and for the displayed parameters. 
\begin{figure}
[t!]
\centering
\includegraphics[width=0.9\columnwidth]{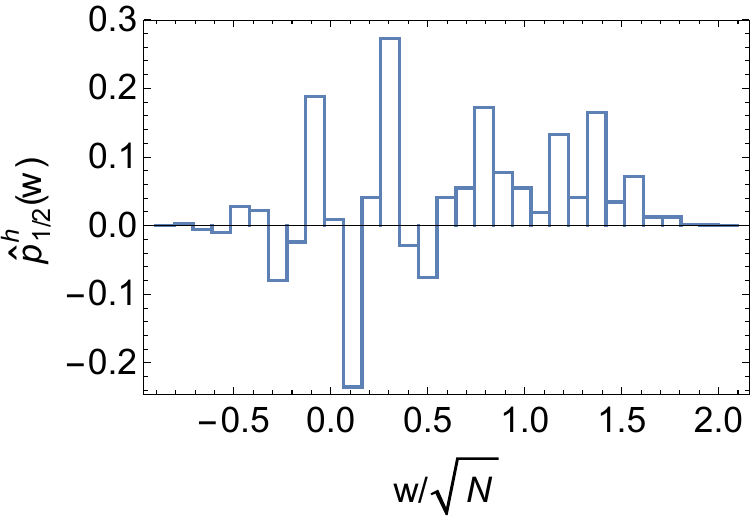}
\caption{The histogram $\hat p^h_{1/2}(w)$ of the quasiprobability distribution $\hat p_{1/2}(w)$. We consider the same values of Fig.~\ref{fig:plot1} and $\alpha=1$. 
}
\label{fig:plot2}
\end{figure}

\section{Conclusions}\label{sec.conclu}
We showed how the quantum advantage is connected to the quasiprobability of the work done in a subinterval of the charging process. If there is some negativity in a quasiprobability ($q=1/2$) of the work done, then there is a quantum advantage.  Thus, the $q=1/2$ quasiprobability plays a privileged role, while for $q\neq 1/2$ negativity can survive although the duration time does not tend to zero.
In particular, for $q=1/2$, if negativity survives in the limit of a large number of cells, the distribution shows super-extensive cumulants, which implies that a quantum advantage is achieved. Basically, under the assumptions used in the proof, the negativity signals non-local interactions in the charging Hamiltonian, giving the quantum advantage in the charging of the battery. For these models, this is related with the generation of quantum coherence and entanglement during time evolution. In contrast, from the point of view of the work statistics, the work always admits a representation ($q=0$) that is a probability distribution (of an observable).
We hope that this can help us to better understand the origin of the quantum advantage with a charging time that tends to zero in the limit of a large number of cells, where quantum effects should be very strong.
Furthermore,  it remains an open problem whether we can check whether the quantum advantage has been achieved by looking at the negativity of a work distribution. In particular, the latter should be measured without destroying the final state of the battery, and a nondemolition diagnostic of the relevant quasiprobability is not supplied here and remains outside the scope of the paper.

\appendix

\section{Leggett-Garg inequality}\label{sec.leggett-garg}
Let us focus on the work $w$ done in the time interval $[t_1,\tau]$ with $t_1\in (0,\tau)$. We can consider the works $w_\tau$ and $w_{t_1}$ done in the time intervals $[0,\tau]$ and $[0,t_1]$. Since $\rho(0)$ is incoherent with respect to the energy basis, these works have the probability distribution $p(w_t)=\sum_{i,k} p_i p_{k|i}\delta(w_t-E^t_k+E^0_i)$ with $t=\tau,t_1$, where $p_i=\bra{E^0_i}\rho(0)\ket{E^0_i}$ and $p_{k|i}=|\bra{E^t_k}U_{t,0}\ket{E^0_i}|^2$. Thus, if $p_q(w)<0$ for some $w$, then there is no joint probability distribution $p(w_\tau,w_{t_1})$, since $w=w_\tau-w_{t_1}$ and thus $p(w_\tau,w_{t_1})$ is such that $\int \delta(w-w_\tau+w_{t_1})p(w_\tau,w_{t_1})dw_\tau dw_{t_1}=p_q(w)$. Furthermore, if there exists a joint probability distribution $p(w_\tau,w_{t_1})$, we obtain the inequality for the covariance
\begin{equation}
|\langle w_\tau w_{t_1}\rangle - \langle w_\tau\rangle \langle w_{t_1}\rangle|\leq \sigma_\tau \sigma_{t_1}\,,
\end{equation}
where $\sigma^2_t=\langle w_t^2\rangle-\langle w_t\rangle^2$. This inequality can be expressed in terms of the variance $\sigma^2=\langle w^2\rangle-\langle w\rangle^2$ as
\begin{equation}\label{eq.LG}
|\sigma^2_\tau + \sigma^2_{t_1}- \sigma^2|\leq 2\sigma_\tau \sigma_{t_1}\,,
\end{equation}
which recalls a Leggett-Garg inequality~\cite{Leggett85,Miller18}. Actually, Leggett-Garg inequalities concern macrorealism
and the existence of suitable multi-time probability distributions, not merely the displayed covariance bound. In particular, Eq.~\eqref{eq.LG} does not depend on $q$ and thus on the quasiprobability representation chosen. Therefore, if the inequality in Eq.~\eqref{eq.LG} is violated, then there is no joint probability $p(w_\tau,w_{t_1})$ and we can get a distribution $p_q(w)$ that is negative for some $w$. It is easy to check that for the charging process under consideration, the inequality is not violated, since $\sigma^2_\tau=0$ and $\sigma^2_{t_1}= \sigma^2$. 
However, this does not exclude the fact that $p_q(w)$ is negative for some $w$. Furthermore, we note that for $q=0$ the quasiprobability distribution $p_0(w)$ is non-negative, so there is a joint probability distribution $p(w_\tau,w_{t_1})$. This explains why the 
inequality is satisfied.

\section{Histogram}\label{sec.histogram}
To determine the histogram $p^h_q(w)$ of the quasiprobability distribution of work $p_q(w)$ from the characteristic function $\chi_q(u)$ we consider the intervals $I_n = [w_n-\Delta w/2,w_n + \Delta w/2]$, where $w_n=n \Delta w$ with $n$ integer. Then, we can determine the histogram by calculating the probability
\begin{equation}
p_n = \int_{I_n} p_q(w) dw = \frac{\Delta w}{2\pi}\int  \chi_q(u)\text{sinc}\left(\frac{u\Delta w}{2}\right) e^{-i u w_n} du
\end{equation}
where $\text{sinc}(x) = \sin(x)/x$. Thus,
\begin{equation}
p^h_q(w) = \sum_n p_n \chi_n(w)\,,
\end{equation}
where $\chi_n(w)=1$ for $w\in I_n$ and $\chi_n(w)=0$ for $w\notin I_n$. To calculate the integral we can focus on the interval $[-2\pi K /\Delta w,2\pi K /\Delta w]$ with $K$ large enough. Of course $p_q(w_n)\approx p_n/\Delta w$ for $\Delta w$ small enough.
In the end, we note that the histogram $\hat p^h_q(w)$ of the rescaled quasiprobability distribution $\hat p_q(w)$ is defined in the same way, and of course has the same form of $p^h_q(w)$ if the support is divided in the same number of bins.

\end{document}